\documentclass[12pt,a4paper]{article}
\usepackage[T1]{fontenc}
\usepackage[utf8]{inputenc}
\usepackage{authblk}
\usepackage{graphicx}
\usepackage[a4paper, total={5.5in, 9in}]{geometry}
\usepackage{float}
\usepackage{caption}
\usepackage{hyperref}
\usepackage{subcaption}
\usepackage{booktabs}
\usepackage{natbib}
\usepackage{url}
\usepackage{xcolor}
\usepackage{color, colortbl}
\definecolor{Gray}{gray}{0.9}
\usepackage{amsmath, amssymb}
\newtheorem{definition}{Definition}
\newtheorem{prop}{Proposition}
\newtheorem{proof}{Proof}

\date{}

\begin{document}

\title{Procrustes-based distances for exploring between-matrices similarity} 

\author[1*]{Angela Andreella}
\author[2]{Riccardo De Santis}
\author[3]{Anna Vesely}
\author[3]{Livio Finos}
\affil[1]{Department of Economics, University Ca' Foscari Venezia, Italy}
\affil[2]{Department of Statistics, University of Padova, Italy}
\affil[3]{Department of Developmental Psychology and Socialization, University of Padova, Italy}
\affil[*]{correspondig author: angela.andreella@unive.it}

\maketitle

\begin{abstract}

The statistical shape analysis called \emph{Procrustes} analysis minimizes the  distance between matrices by similarity transformations. The method returns a set of optimal orthogonal matrices, which project each matrix into a common space. This manuscript presents two types of distances derived from \emph{Procrustes} analysis for exploring between-matrices similarity. 
The first one focuses on the residuals from the \emph{Procrustes} analysis, i.e., the \emph{residual-based} distance metric. In contrast, the second one exploits the fitted orthogonal matrices, i.e., the \emph{rotational-based} distance metric. Thanks to these distances, similarity-based techniques such as the multidimensional scaling method can be applied to visualize and explore patterns and similarities among observations. The proposed distances result in being helpful in functional magnetic resonance imaging (fMRI) data analysis. 
The brain activation measured over space and time can be represented by a matrix. The proposed distances applied to a sample of subjects -- i.e., matrices -- revealed groups of individuals sharing patterns of neural brain activation.

{\bf Keywords:} Procrustes method, ProMises model, Orthogonal transformation, similarity, fMRI group analysis, fMRI data
\end{abstract}

\section{Introduction}
Applications in several fields, such as ecology \citep{saito2015should}, biology \citep{rohlf1990extensions}, analytical chemometrics \citep{andrade2004procrustes}, psychometrics \citep{green1952orthogonal,mccrae1996evaluating}, and neuroscience \citep{Haxby} need to compare information described by matrices expressed in an arbitrary coordinate system. The dimension of the matrices corresponding to this arbitrary coordinate system results to be a functional misalignment. In this context, the statistical shape analysis \citep{dryden2016statistical} called \emph{Procrustes} analysis \citep{Gower} can be helpful. Briefly, the \emph{Procrustes} analysis aligns the matrices into a common reference space by similarity transformations (i.e., rotation, reflection, translation, and scaling transformations). The optimal similarity transformations are those that minimize the squared  distance between the matrices.

Several \emph{Procrustes}-based functional alignment approaches can be found in the literature; two of the most used ones are the orthogonal \emph{Procrustes} problem (OPP) \citep{Jos} and the generalized \emph{Procrustes} analysis (GPA) \citep{gower1975generalized}. The first deals with the alignment of two matrices, while the second finds optimal similarity transformations when more than two matrices are analyzed. OPP has a closed-form solution, while GPA is based on an iterative algorithm proposed by \cite{gower1975generalized}. Since the \emph{Procrustes} problem can be seen as a least squares problem, \cite{Goodall} translated it into a statistical model, i.e., the perturbation model, where the error terms follow a matrix normal distribution \citep{gupta2018matrix}. 

Neuroscience is one of the fields where \emph{Procrustes}-based methods are most widely used. In particular, functional magnetic resonance imaging (fMRI) is the most widely used technique for studying the neural underpinnings of human cognition. Brain activation is expressed as the correlation between the sequence of cognitive stimuli and the sequence of measured blood oxygenation levels (BOLD). In response to neural activity, changes in brain hemodynamics affect the local intensity of the magnetic resonance signal, that is, the voxel intensity (single-volume elements). However, various criticalities arise when analysis (e.g., classification analysis, inference analysis) between subjects is performed. The anatomical and functional structures of the brain greatly vary between subjects, even if time-synchronized stimuli are proposed to the participants \citep{Watson, Tootell, Hasson}. For that, the alignment step is an essential part of the preprocessing procedure in fMRI group-level analysis. Anatomical normalization (e.g., \cite{talairach1988co, jenkinson2002improved, Fischl}) fixes the anatomical misalignment through affine transformations, where brain images are aligned to a standard anatomical template (e.g., Talairach template \citep{talairach1988co}, Montreal Imaging Institute (MNI) template \citep{collins1994automatic}). However, the anatomical alignment does not capture the functional variability between subjects, which is a well-known problem in the neuroscience literature \citep{Watson, Hasson, Tootell}. 

The brain activation of one subject can be described by a matrix where the rows represent the time points/stimuli and the columns the voxels. Therefore, each row shows the response activation to one stimulus across all voxels, and each column expresses the time series of activation for each voxel. The functional misalignment can be focused on the columns between matrices, i.e., the time series of activations are not in correspondence between subjects, while the response activations are since the stimuli are generally time-synchronized \citep{Haxby, andreella2022procrustes}.
In the context of fMRI data, one of the most popular \emph{Procrustes}-based functional alignment methods is the \emph{hyperalignment} technique proposed by \cite{Haxby}, which is a sequential approach to OPP. However, both OPP and GPA and \emph{hyperalignment} suffer from low interpretability of aligned matrices (i.e., fMRI images) and related results (e.g., statistical t-tests, classifier coefficients) as well as in-applicability in high-dimensional data. In particular, in fMRI data analysis, the first problem leads to losing the anatomical interpretation of the final aligned images, and the second one makes it impossible to apply the alignment method to the whole brain. The low interpretability is caused by the ill-posed structure of the \emph{Procrustes}-based approaches: they do not return a unique solution for the optimal orthogonal transformation. For further details about the functional alignment problem in the fMRI data analysis framework, please see \cite{andreella2022enhanced}.

For that, \cite{andreella2022procrustes} proposed an extension, i.e., the ProMises model, of the perturbation model developed by \cite{Goodall}. In particular, the perturbation model rephrases the \emph{Procrustes} problem as a statistical model. The extension of \cite{andreella2022procrustes} is focused on inserting a penalization in the orthogonal matrix's estimation process, specifying a proper prior distribution for the orthogonal matrix parameter. The von Mises-Fisher distribution \citep{downs1972orientation} is used to insert prior information about the final structure of the common space. Thanks to that, the no-uniqueness problem of the \emph{Procrustes}-based methods is solved, getting an interpretable estimator for the orthogonal matrix transformations. This permits to have unique aligned matrices as well as related statistical inference results. The alignment process does not affect the type I error since the ProMises model can be seen as a procedure that sorts the null hypotheses based on a priori information \citep{andreella2022valid}.
The computation of the maximum a posteriori estimate is straightforward; in fact, the von Mises-Fisher distribution is a conjugate prior to the matrix normal distribution \citep{gupta2018matrix}, which is the distribution of the error terms in the ProMises and perturbation models.

In this work, we present a method that exploits the information coming from the functional misalignment resulting from \emph{Procrustes}-based methods (e.g., GPA, \emph{hyperlalignment} and ProMises model). We propose here two distance metrics \citep{deza2006dictionary} that capture different perspectives of similarity/dissimilarity between matrices, e.g., subjects in the fMRI cases. The minimization problem solved by \emph{Procrustes}'s methods can also be defined as distance among objects \citep{dryden2016statistical}. The first distance metric presented here is based on the residuals coming from the solution of a Procrustes problem. The \emph{residual-based} distance expresses then how the matrices/subjects are different/similar after functional alignment. In this case, the distance metric captures the dissimilarity/similarity in terms of noise since the matrices have the same orientations after functional alignment. Instead, the second distance exploits the orthogonal matrix parameters solution of the \emph{Procrustes} problem. The \emph{rotational-based} distance computes the squared  distance between these estimated orthogonal matrices. As we will see, this metric measures the level of dissimilarity/similarity in orientation between matrices/subjects before functional alignment.

In the paper, we show how these metrics can be used in distance-based techniques such as the multidimensional scaling method \citep{carroll1998multidimensional}, hierarchical clustering \citep{murtagh2012algorithms} and t-distributed stochastic neighbor embedding (t-SNE) \citep{van2008visualizing} in order to visualize and quantify patterns and shared characteristics between matrices (i.e., individuals described by multiple dimensions). 

The paper is organized as follows. Section
\ref{promisesmodel} introduces the \emph{Procrustes}-based methods. The core of the manuscript is contained in Section \ref{rotations}, where the \emph{residual-based} and \emph{rotation-based} distances are proposed. Finally, we explain how to use the distances between rotations and residuals as a tool to understand the underlying clusters between subjects in the fMRI data analysis framework in Section \ref{application}. The analyses of this manuscript are performed using the \texttt{R} package \texttt{alignProMises} available at \url{https://github.com/angeella/alignProMises} for the functional alignment part, and using the \texttt{R} package \texttt{rotoDistance} available at \url{https://github.com/angeella/rotoDistance} for the computation of the \emph{rotational-based} and \emph{residual-based} distances.

\section{Procrustes analysis}\label{promisesmodel}

Let $\{\boldsymbol{X}_i \in \mathbb{R}^{n \times m}\}_{i = 1,\dots, N}$ be a set of matrices to be aligned. The Procrustes analysis uses similarity transformations \citep{gower1975generalized}, i.e., scaling, rotation/ reflection, and translation, to map $\{\boldsymbol{X}_i \in \mathbb{R}^{n \times m}\}_{i = 1,\dots, N}$ into a common reference space.

If only two matrices are analyzed, i.e., $N=2$, we can consider one of the two matrices as a common reference matrix. The orthogonal \emph{Procrustes} problem (OPP) is then applied and defined as:
\begin{equation}\label{eq:opp}
	\min_{\boldsymbol{R}_i, \alpha_i, t_i} || \alpha (X_i - 1_n t) R - X_j||_{F}^2 \quad \text{subject to } R_i \in \mathcal{O}(m)
\end{equation}
where $\mathcal{O}(m)$ is the orthogonal group in dimension $m$, $|| \cdot||_{F}$ is the Frobenius norm, $\alpha \in \mathbb{R}^{+}$ is the isotropic scaling, $t \in \mathbb{R}^{1 \times m}$
defines the translation vector, and $1_n \in \mathbb{R}^{1 \times n}$ is a vector of ones.

The optimal translation results to be the column-centering, while $\boldsymbol{R}$ and $\alpha$ equal
\begin{equation}\label{eq:svd}
	\hat{R} = \boldsymbol{U} \boldsymbol{V}^\top; \quad \hat{\alpha}_{\hat{R}} = \frac{\text{tr}(D)}{||  \hat{R}^\top X_i^\top||^2_{\text{F}}}
\end{equation} 
where $U D V^\top$ is the singular value decomposition of $\boldsymbol{X}_i^\top  \boldsymbol{X}_j$. 

If more than two matrices are analyzed, i.e., $N>2$, the generalized \emph{Procrustes} analysis (GPA) must be applied. In this case, the set of matrices $\{\boldsymbol{X}_i \in \mathbb{R}^{n \times m}\}_{i = 1,\dots, N}$ are mapped by similarity transformations into a common reference matrix $\boldsymbol{M} \in \mathbb{R}^{n \times m}$. This common reference matrix can be defined in several ways, e.g., element-wise arithmetic mean. The GPA is defined as

\begin{equation}\label{eq:gpa}
	\min_{\boldsymbol{R}_i, \alpha_i, t_i} \sum_{i=1}^{N} || \alpha_i (X_i - 1_n t_i) R_i - M||_{F}^2 \quad \text{subject to } R_i \in \mathcal{O}(m).
\end{equation} 

Unlike OPP, GPA does not have a closed-form solution for $R_i$ and $\alpha_i$, and an iterative algorithm must be used where at each step, the reference matrix is updated \citep{gower1975generalized}.

Another approach is the perturbation model proposed by \cite{Goodall}, where the least squares problem defined in Equation \ref{eq:gpa} is translated as a statistical model assuming that $\{\boldsymbol{X}_i\}_{i = 1,\dots, N}$ are noisy rotations of a common space $\boldsymbol{M}$.

The perturbation model is then defined as follows:
\begin{equation}\label{promises}
	\boldsymbol{X}_i = \alpha_i(\boldsymbol{M} + \boldsymbol{E}_i) \boldsymbol{R}_i^\top +\boldsymbol{1}^\top_n t_i
\end{equation} 
where $\boldsymbol{E}_i$ is the random error matrix following a normal matrix distribution \citep{gupta2018matrix} $\boldsymbol{E}_i \sim \mathcal{MN}_{nm}(0, \boldsymbol{\Sigma}_n, \boldsymbol{\Sigma}_m)$, with $\boldsymbol{\Sigma}_n \in \mathbb{R}^{n \times n}$ and $\boldsymbol{\Sigma}_m \in \mathbb{R}^{m \times m}$.  The similarity transformations are represented by the following parameters $\boldsymbol{R}_i$, $\alpha_i$, and $t_i$ that must be estimated for each $i = 1, \dots, N$. The optimal similarity transformations $\hat{R}_i$ and $\hat{\alpha_i}_{\hat{R}_i}$ are slight modifications of the ones found by OPP and GPA:

\begin{equation}\label{eq:svd2}
	\hat{R}_i\ = \{\boldsymbol{U}_i \boldsymbol{V}_i^\top\}_{i = 1,\dots, N}; \quad \hat{\alpha_i}_{\hat{R}_i} = \frac{\text{tr}(D_i)}{|| \Sigma_m^{-1/2} \hat{R}_i^\top X_i^\top \Sigma_n^{-1/2}||^2_{\text{F}}} \quad \forall i \in \{1,\dots, N\}
\end{equation} 
where $U_i D_i V_i^\top$ is the singular value decomposition of $\boldsymbol{X}_i^\top \Sigma_n^{-1}  \boldsymbol{X}_j \Sigma_m^{-1}$. 

The extension of the perturbation model is proposed by \cite{andreella2022procrustes}, where the orthogonal matrix parameter $\boldsymbol{R}_i$ follows a von Mises-Fisher distribution \citep{downs1972orientation}:

\begin{equation*}
	f(\boldsymbol{R}_i) \sim C(\boldsymbol{F},k) \exp(k \boldsymbol{F} \boldsymbol{R}_i)
\end{equation*}
where $\boldsymbol{F} \in \mathbb{R}^{m \times m}$ is the location matrix parameter, $k \in \mathbb{R}^{+}$ represents the regularization parameter and $C(\boldsymbol{F},k)$ is the normalizing constant. \cite{andreella2022procrustes} found that the maximum a posteriori estimates are slight modifications of the perturbation model proposed by \cite{Goodall} (i.e., without imposing the von Mises-Fisher prior distribution for $\boldsymbol{R}_i$). The estimators for the sets of parameters $\{\boldsymbol{R}_i\}_{i = 1,\dots, N}$ and $\{\boldsymbol{\alpha}_i\}_{i = 1,\dots, N}$ are essentially the same but decomposing $\boldsymbol{X}_i^\top \Sigma_n^{-1} \boldsymbol{M} \Sigma_m^{-1} + k \boldsymbol{F}$ instead of $\boldsymbol{X}_i^\top \Sigma_n^{-1} \boldsymbol{M} \Sigma_m^{-1}$. The straightforward solutions are due to the conjugacy of the von Mises–Fisher distribution to the matrix normal distribution \citep{green2006bayesian, andreella2022procrustes}. Therefore, the prior information enters directly into the singular value decomposition step of the estimation process. 

The motivation to impose an a priori distribution to the orthogonal matrix parameter $\boldsymbol{R}_i$ stems from the assumption that ``the anatomical alignment is not so far from the truth". The information of the three-dimensional spatial coordinates of the voxels is then inserted into the estimation process thanks to a proper definition of the prior location parameter $\boldsymbol{F} \in \mathbb{R}^{m \times m}$. \cite{andreella2022procrustes} define $\boldsymbol{F}$ as a similarity Euclidean distance. In this way, the rotation loadings that combine closer voxels are higher than the ones that combine voxels that are far apart. In addition, defining $\boldsymbol{F}$ as a similarity Euclidean matrix leads to $\boldsymbol{X}_i^\top \Sigma_n^{-1} \boldsymbol{M} \Sigma_m^{-1} + k \boldsymbol{F}$ having full rank, i.e., unique solution for $\boldsymbol{R}_i$.  

Finally, \cite{andreella2022procrustes} proposed an efficient version of the ProMises model in the case of high-dimensional data. The problem when $m >> n$ arises since the ProMises model, and also the perturbation model, must compute $N$ singular value decompositions of matrices with dimensions $m \times m$. \cite{andreella2022procrustes} use specific semi-orthogonal transformations to project the matrices $\boldsymbol{X}_i \in \mathbb{R}^{n \times m}$ into the lower dimensional space $\mathbb{R}^{n \times n}$. In particular, if we consider as $m\times n$ semi-orthogonal transformation $\boldsymbol{Q}_i$ the ones coming from the thin singular value decomposition \citep{bai2000templates} of $\boldsymbol{X}_i$ we reach the same fit of data but reducing the time complexity from $\mathcal{O}(m^3)$ to $\mathcal{O}(m n^2)$, and the space complexity from $\mathcal{O}(m^2)$ to $\mathcal{O}(mn)$.

Briefly, the Efficient ProMises applies the semi-orthogonal transformation $\boldsymbol{Q}_i$ to $\boldsymbol{X}_i$ and then applies the ProMises model on the set of lower dimensional matrices $\{\boldsymbol{X}_i \boldsymbol{Q}_i \in \mathbb{R}^{n \times n}\}$. The efficient ProMises model allows the alignment of high-dimensional data such as fMRI data where the dimension $m$ (i.e., the number of voxels) equals approximately $200,000$.

For further details about the ProMises model and its Efficient version, please see \cite{andreella2022procrustes}. 

\section{Procrustes-based distances}\label{rotations}

\emph{Procrustes}-based methods (i.e., OPP, GPA, perturbation model, or ProMises model) find the orthogonal matrices that, applied to the original matrices, minimize the Frobenius distance among resulting matrices. It is, therefore, natural to define a distance that is based on this quantity: the squared residuals among aligned matrices. In this case, we measure how different two matrices are beyond rotation. Two matrices can look very different, while they may result to be very similar after rotation. \emph{Residual-based} distance succeeds in capturing this aspect, thus evaluating only the distance between matrices net of rotations.

The second kind of distance that we will define is based on the \emph{rotational effort} that is taken to align one matrix $\boldsymbol{X}_i$ to another matrix $\boldsymbol{X}_j$. This effort is measured as the distance between the orthogonal matrix that solves the Procrustes problem  $\hat{\boldsymbol{R}}_i$ and $\boldsymbol{I}_m$ (i.e., the matrix that does not operate any rotation): the larger the distance between the $\hat{\boldsymbol{R}}_i$ and $\boldsymbol{I}_m$, the bigger the effort to align $\boldsymbol{X}_i$ to $\boldsymbol{X}_j$.

In the following, we give the formal definitions of \emph{residual-based} and \emph{rotational-based} distances:

\begin{definition}\label{distance2}
	Consider a set of matrices $\{\boldsymbol{\hat{X}}_i \in \mathbb{R}^{n \times m}\}_{i=1, \dots, N}$ functionally aligned by some \emph{Procrustes}-based method presented in Section \ref{promisesmodel}, i.e.,
	\begin{equation*}
		\boldsymbol{\hat{X}}_i = \hat{\alpha_i}_{\boldsymbol{\hat{R}_i}}\boldsymbol{X}_i \boldsymbol{\hat{R}_i}.
	\end{equation*}
	The \emph{residual-based} distance is defined as:
	\begin{equation}\label{eq:res}
		d_{Re}(\boldsymbol{\hat{X}}_i, \boldsymbol{\hat{X}}_j) = || \boldsymbol{\hat{X}}_i  - \boldsymbol{\hat{X}}_j||_{F}^2.
	\end{equation}
\end{definition}

We can note that the \emph{residual-based} distance defined in Equation \ref{eq:res} is directly related to the GPA defined in Equation \ref{eq:gpa}. If we consider two matrices, the distance is simply the pair's contribution within the GPA minimization problem, precisely the optimization's residuals. 

\begin{definition}\label{distance}
	Consider a set of orthogonal matrices $\{\hat{\boldsymbol{R}}_i \in \mathcal{O}(m)\}_{i=1, \dots, N}$ estimated by some \emph{Procrustes}-based method presented in Section \ref{promisesmodel}. The  \emph{rotational-based} distance is defined as:
	\begin{equation}
		d_{Ro}(\hat{\boldsymbol{R}}_i, \hat{\boldsymbol{R}}_j) = || \hat{\boldsymbol{R}}_i - \hat{\boldsymbol{R}}_j||_{F}^2.
	\end{equation}
\end{definition}

Since both distances are based on the matrix Frobenius norm, this implies that $d_{Re} (\cdot)$ and $d_{Ro} (\cdot)$ can be considered directly as a valid metric, i.e., distance functions $d_{Re}: \mathbb{R}^{n \times m} \times \mathbb{R}^{n \times m} \rightarrow \mathbb{R}^{\ge 0}$ and $d_{Ro}: \mathcal{O}(m)\times \mathcal{O}(m) \rightarrow \mathbb{R}^{\ge 0}$. 





Therefore, if $d_{Re}=0$, the two matrices are functionally similar without considering the orientation characteristics. In the same way, as $d_{Re}$ increases, dissimilarity in functional terms increases without considering orientation again. Instead, if $d_{Ro}=0$, we have two images sharing the same orientation, i.e., functional (mis)alignment concerning the reference matrix $\boldsymbol{M}$. In the same way, if $d_{Ro} >0$, the two matrices have different orientations in terms of column dimension.


Indeed, the definition of \emph{rotational-based} distance can be significantly simplified, thus simplifying both the computational calculation and the interpretation of distance itself. This is formalized in the following:
\begin{prop}
	The \emph{rotational-based} distance defined in Definition \ref{distance} can be expressed as:
	\begin{equation*}
		d_{Ro}(\hat{\boldsymbol{R}}_i, \hat{\boldsymbol{R}}_j) = 2m -2tr( \hat{\boldsymbol{R}}_i^\top \hat{\boldsymbol{R}}_j)
	\end{equation*}
	and takes values in $[0, 4m]$. The same result can be obtained using the \emph{residual-based} distance $d_{Re}$ defined in Equation \ref{eq:res} when $\boldsymbol{\hat{X}}_i, \boldsymbol{\hat{X}}_j \in \mathcal{O}(m)$.
\end{prop}
\begin{proof}
	\begin{align*}
		d(\hat{\boldsymbol{R}}_i, \hat{\boldsymbol{R}}_j) &= ||\hat{\boldsymbol{R}}_i - \hat{\boldsymbol{R}}_j||_{F}^2 = tr[(\hat{\boldsymbol{R}}_i  - \hat{\boldsymbol{R}}_j)^\top (\hat{\boldsymbol{R}}_i  - \hat{\boldsymbol{R}}_j)] \\
		&= tr(\hat{\boldsymbol{R}}_i^\top \hat{\boldsymbol{R}}_i  -  \hat{\boldsymbol{R}}_i^\top \hat{\boldsymbol{R}}_j - \hat{\boldsymbol{R}}_j^\top \hat{\boldsymbol{R}}_i  + \hat{\boldsymbol{R}}_j^\top \hat{\boldsymbol{R}}_j) = 2m -2tr( \hat{\boldsymbol{R}}_i^\top \hat{\boldsymbol{R}}_j).
	\end{align*}
	Considering the \emph{rotational}-based distance, the trace of the product between the two orthogonal matrices $\hat{\boldsymbol{R}}_i$ and $\hat{\boldsymbol{R}}_j$ can take only values between $m$ and $-m$ since the eigenvalues of an orthogonal matrix lie on the unit circle (i.e., they have module equal to $1$). If $tr( \hat{\boldsymbol{R}}_i^\top \hat{\boldsymbol{R}}_j) = m$, this means that $ \hat{\boldsymbol{R}}_i =  \hat{\boldsymbol{R}}_j$ since the only way for an orthogonal matrix to have all the eigenvalues equal to $1$ is being an identity matrix. In the same way, if the trace equals $-m$, this means that $ \hat{\boldsymbol{R}}_i^\top \hat{\boldsymbol{R}}_j = - I_m$, i.e., $ \hat{\boldsymbol{R}}_i =  - \hat{\boldsymbol{R}}_j$.
	
\end{proof}


The \emph{residual-based} distance and the \emph{rotational-based} distance are then computed for each pair of aligned matrices $\{\hat{\boldsymbol{X}}_i \}_{i = 1, \dots, N}$ and for each pair of orthogonal matrices $\{\hat{\boldsymbol{R}}_i \}_{i = 1, \dots, N}$, resulting in a the global distance matrix $D \in \mathbb{R}^{N \times N}$. Information from different matrices with large dimensions can be summarized through the proposed distance matrix, which will turn out to be of lower dimension, i.e., of dimension $N \times N$. These distance matrices can be handy in various applications, particularly when handling big data. In the literature, various statistical methods are based on the distance matrix. However, they generally focus on analyzing the distances of several covariates described by a single matrix or on analyzing multiple distance matrices (e.g., the INDSCAL method proposed by \cite{carroll1998multidimensional}). In contrast, the distance matrix proposed in this manuscript directly summarizes several large matrices' similarity and dissimilarity characteristics.

This matrix $D$ can then be used inside a dissimilarity-based algorithm such as the multidimensional scaling technique \citep{carroll1998multidimensional}, hierarchical clustering \citep{murtagh2012algorithms} and t-distributed stochastic neighbor embedding (t-SNE) \citep{van2008visualizing}. 

\section{Application}\label{application}

We analyze $24$ subjects passively looking at food and no-food (office utensils) images collected by \cite{smeets2013allured}. The food/no-food images are proposed to the participants alternately (24 seconds of food images and 24 seconds of no-food images) with a rest block of 12 seconds on average showing a crosshair. The food stimulus is a collection of attractive foods to capture brain activations concerning self-regulation in response to viewing images of tempting (i.e., palatable high-caloric) food \citep{smeets2013allured}.


The dataset was preprocessed using the Functional MRI of the Brain Software Library (FSL) \citep{jenkinson2012fsl} following a standard processing pipeline. The registration step to standard space images was computed using FLIRT \citep{jenkinson2001global}, the motion correction using MCFLIRT \citep{jenkinson2002improved}, the non-brain removal using BET \citep{jenkinson2002improved}, and spatial smoothing using a Gaussian Kernel FWHM ($6$mm). Finally, the intensity normalization of the entire four-dimensional dataset was computed by a single multiplicative factor, and the high-pass temporal filtering (Gaussian-weighted least-squares straight line fitting, with sigma=$64.0$s) was applied. The raw dataset is available at \url{https://openneuro.org/datasets/ds000157/versions/00001}, while the preprocessed one is available in the \texttt{R} package \texttt{rotoDistance} (\url{https://github.com/
	angeella/rotoDistance}). For details about the experimental design and data acquisition, please see \cite{smeets2013allured}.

We analyze the right calcarine sulcus composed of $237$ voxels being an area involved in processing visual information and related to regions involved in the regulation of food intake \citep{smeets2013allured}. However, the whole brain can be analyzed instead of only a region of interest (e.g., right calcarine sulcus). In fact, the distances proposed permit to resume complex high-dimensional data, like the fMRI ones, that are generally composed by $N$ matrices having dimensions $300 \times 200,000$ (i.e., $300$ time points and $200,000$ voxels), through a matrix of low $N\times N$ dimensions.

The ProMises model is fitted on preprocessed data, and aligned images are then used to compute the distance matrix $\boldsymbol{D}_{Re} \in \mathbb{R}^{24 \times 24}$ (i.e., \emph{residual}-based distances), while 
the corresponding optimal rotation matrices are used to compute $\boldsymbol{D}_{Ro} \in \mathbb{R}^{24 \times 24}$ (i.e., \emph{rotational}-based distances) as described in Section \ref{rotations}. 

These two types of \emph{Procrustes}-based distance capture different information, i.e., the between-subjects dissimilarity in terms of brain activations before and after functional alignment. Figure \ref{fig:scatterplot} shows the distances $\boldsymbol{D}_{Re}$ and $\boldsymbol{D}_{Ro}$ for each pair of subjects. The correlation between them is very low, i.e., $\approx 0.03$, as we can note from Figure \ref{fig:scatterplot}, i.e.,  $\boldsymbol{D}_{Re}$ and $\boldsymbol{D}_{Ro}$ return two distinct insights regarding the between-subjects dissimilarity brain activations.

\begin{figure}
	\centering
	\includegraphics[width=.8\textwidth]{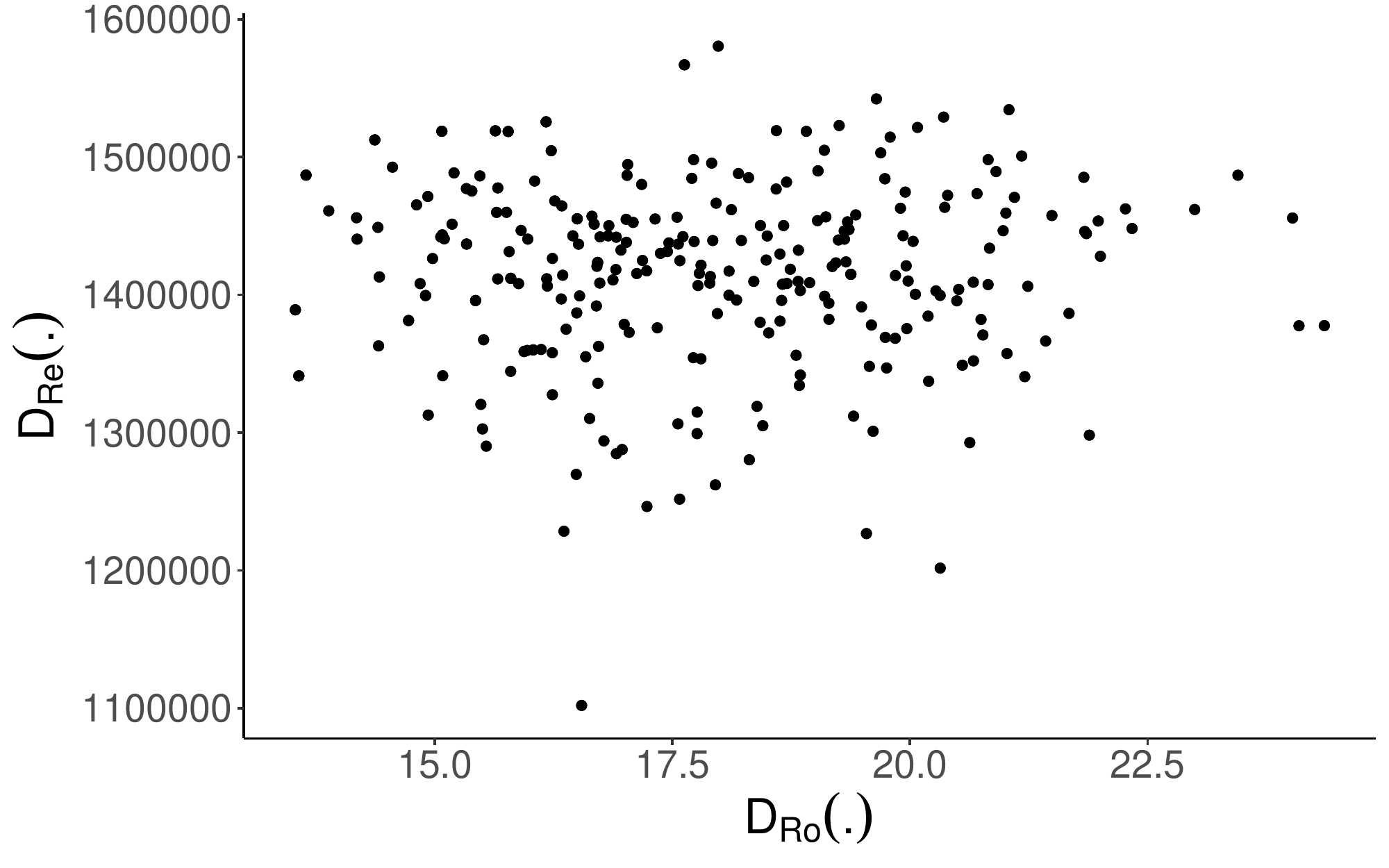}
	\caption{Scatterplot between $\boldsymbol{D}_{Re}$ and $\boldsymbol{D}_{Ro}$.}
	\label{fig:scatterplot}
\end{figure}

The multidimensional scaling (MDS) technique \citep{carroll1998multidimensional} is now applied considering $\boldsymbol{D}_{Ro}$ as distance matrix. A second analysis is run using $\boldsymbol{D}_{Re}$. For comparison porposes, we also computed the Euclidean distances between images that are not functionally aligned. We denote the corresponding distance matrix as $D_{raw}(\boldsymbol{X}_i, \boldsymbol{X}_j) = || X_i - X_j||_{F}^2$. We used the \texttt{smacof} \texttt{R} package \citep{de2009multidimensional} for applying the multidimensional scaling technique. We decided to apply the spline MDS (monotone spline transformation) with as much flexibility as possible. See \cite{de2009multidimensional} for more details. 

Furthermore, we have some covariates for each subject to analyze, briefly described in Table \ref{tab:cov} together with age, body mass index (BMI), and other information. We then analyze these covariates with the matrix of fitted configurations computed by the multidimensional scaling approach. Please see \cite{smeets2013allured} for more details.

\begin{table}[]
	\centering
	\begin{tabular}{l|l}
		\textbf{COVARIATE} & \textbf{DESCRIPTION} \\
		\toprule
		Diet importance    &  Importance to dieting on a $5$ point scale\\
		\rowcolor[gray]{.9}
		Diet success    &  Success in the diet on a $5$ point scale\\
		Appetite pre-experiment   & Appetite before the scan on a $5$ point scale with $3$ items\\
		&(Cronbach’s $\alpha= 0.84$)\\
		\rowcolor[gray]{.9}
		
		Appetite post-experiment     &  Appetite after the scan on a $5$ point scale with $3$ items \\
		\rowcolor[gray]{.9}
		
		&(Cronbach’s $\alpha= 0.91$)\\
		Cycle phase     & Date of their last menstrual period \\
		&(follicular and ovulation, luteal and menstrual phases)
	\end{tabular}
	\caption{Description of the covariates concerning the dataset from \citep{smeets2013allured}}
	\label{tab:cov}
\end{table}

Focusing firstly on the distance matrix $D_{Ro}$, Figure \ref{fig:stress} shows the stress value considering several numbers of dimensions $K = \{1, \dots, 20\}$ into the multidimensional scaling method. We evaluated that $K = 11$ is a good value corresponding to stress $\approx 0.05$.

\begin{figure}
	\centering
	\includegraphics[width=.8\textwidth]{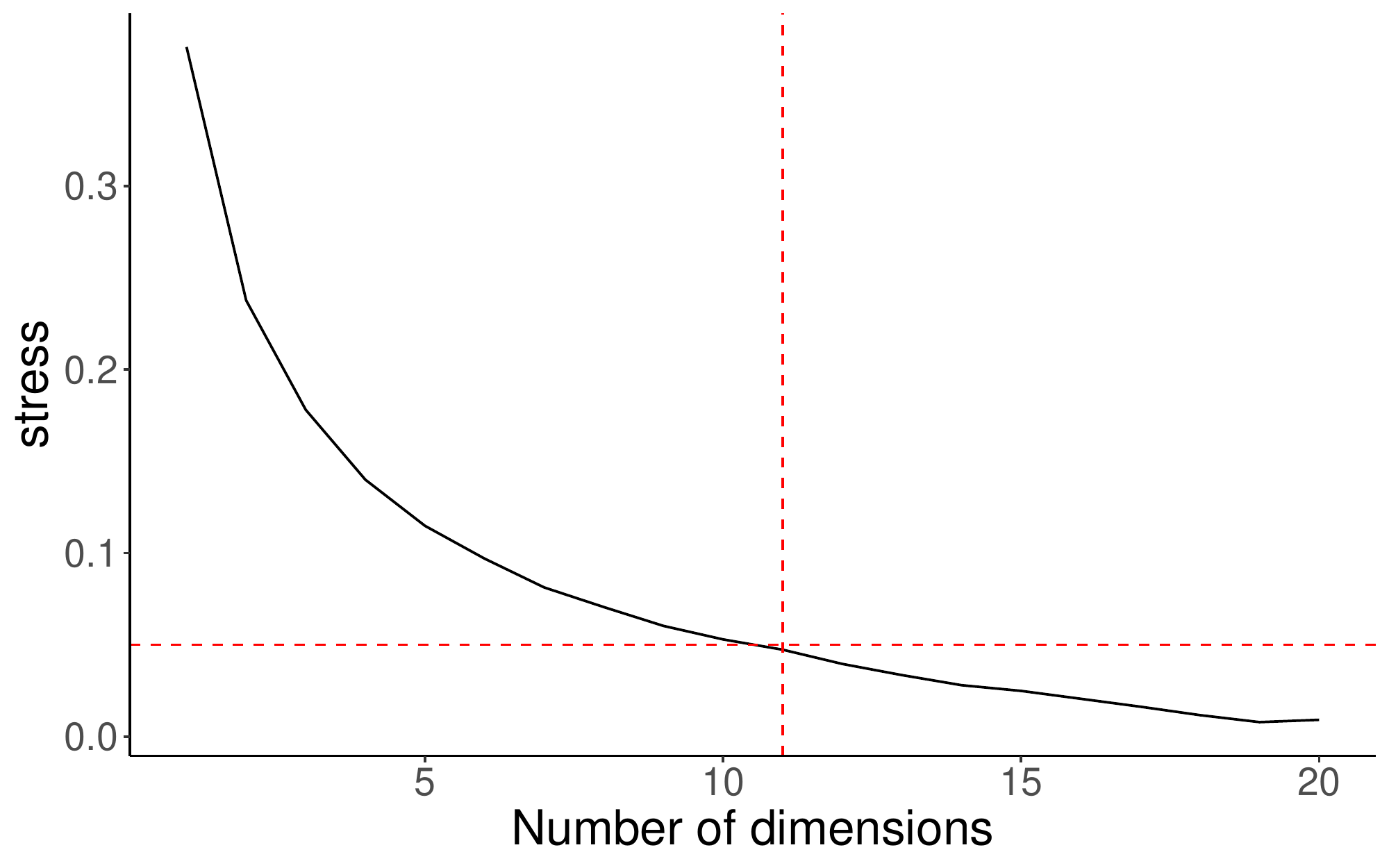}
	\caption{$D_{Ro}$ analysis: Stress values considering several numbers of dimensions in the multidimensional scaling method. The dotted red lines refer to stress equal to $0.05$, and a number of dimensions equal to $11$.}
	\label{fig:stress}
\end{figure}

We then performed simple generalized linear regressions with the covariates as dependent variables and the $11$ configurations fitted by MDS as explanatory variables. We found a significant relationship between the covariate diet success and $6$ dimension ($t = -4.171$, $p = 0.0065$) in accordance with the results found by \cite{smeets2013allured}.
The p-value reported was adjusted for multiple testing using the Bonferroni method \citep{goeman2014multiple}. Therefore, Figure \ref{fig:distance} shows the $1$ and $6$ fitted configurations along with the main covariate (diet success) and cycle phase covariate.

\begin{figure}
	\centering
	\includegraphics[width=\textwidth]{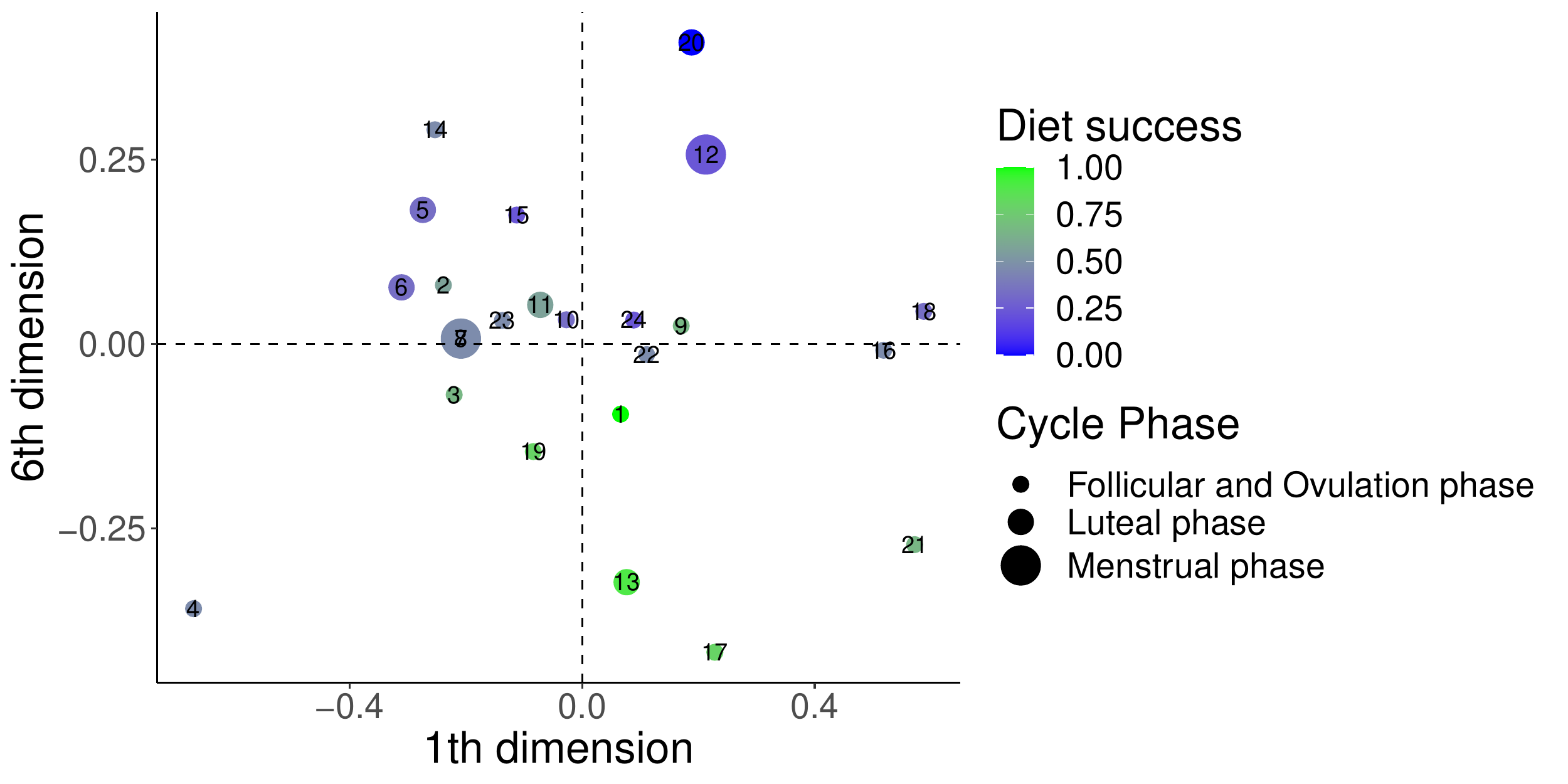}
	\caption{$D_{Ro}$ analysis: The $x$-axis represents the $1$th fitted configuration computed by the multidimensional scaling approach, while the $y$-axis shows the $6$th fitted configuration. The color gradient describes the diet success covariate (scaled), while the size of the points specifies the phase of the cycle.}
	\label{fig:distance}
\end{figure}

At first look, we can note how the $y$ axis represents the success in a diet, where negative values correspond to low success and positive values high success in a diet. We can also note, for example, that subjects $16$ and $18$ share the same functional misalignment with similar diet success values and the same cycle phase.

However, if we instead apply multidimensional scaling on the matrix of \emph{residual-based} distances, we do not find patterns as clear as those found using \emph{rotation-based} distances, as can be seen from Figure \ref{fig:distance_res}. In this case, we automatically set the number of dimensions equal to $11$ (which is equivalent to a stress of $0.05$). The generalized linear regressions did not show any significant features, unlike the first analysis based on the distances of the rotations.

\begin{figure}
	\centering
	\includegraphics[width=\textwidth]{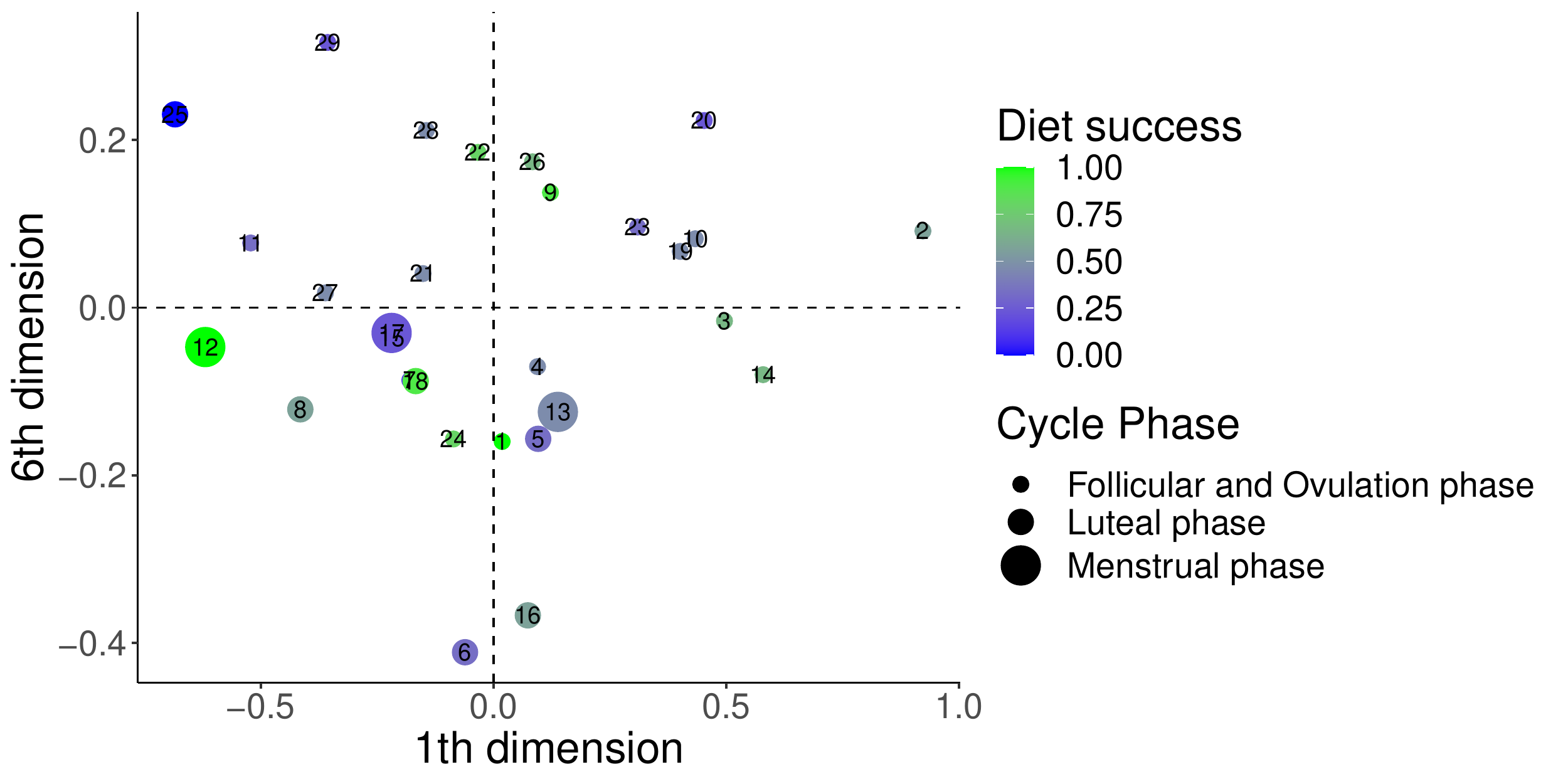}
	\caption{$D_{Re}$ analysis: The $x$-axis represents the $1$th fitted configuration computed by the multidimensional scaling approach, while the $y$-axis shows the $6$th fitted configuration. The color gradient describes the appetite post covariate (scaled) while the size of the points specifies the phase of the cycle. 
}
\label{fig:distance_res}
\end{figure}

Finally, we would have the same situation found using $D_{Re}$ (or even worse) if we used as distance matrix $D_{raw}$ (i.e., using images not functionally aligned). The stress value equals $0.02$ considering $11$ dimensions, and no significant dimensions were found from the generalized linear regressions. Figure \ref{fig:distance_raw} represents the multidimensional scaling results using $D_{raw}$. The two dimensions do not capture the subject-level features analyzed. 

\begin{figure}
\centering
\includegraphics[width=\textwidth]{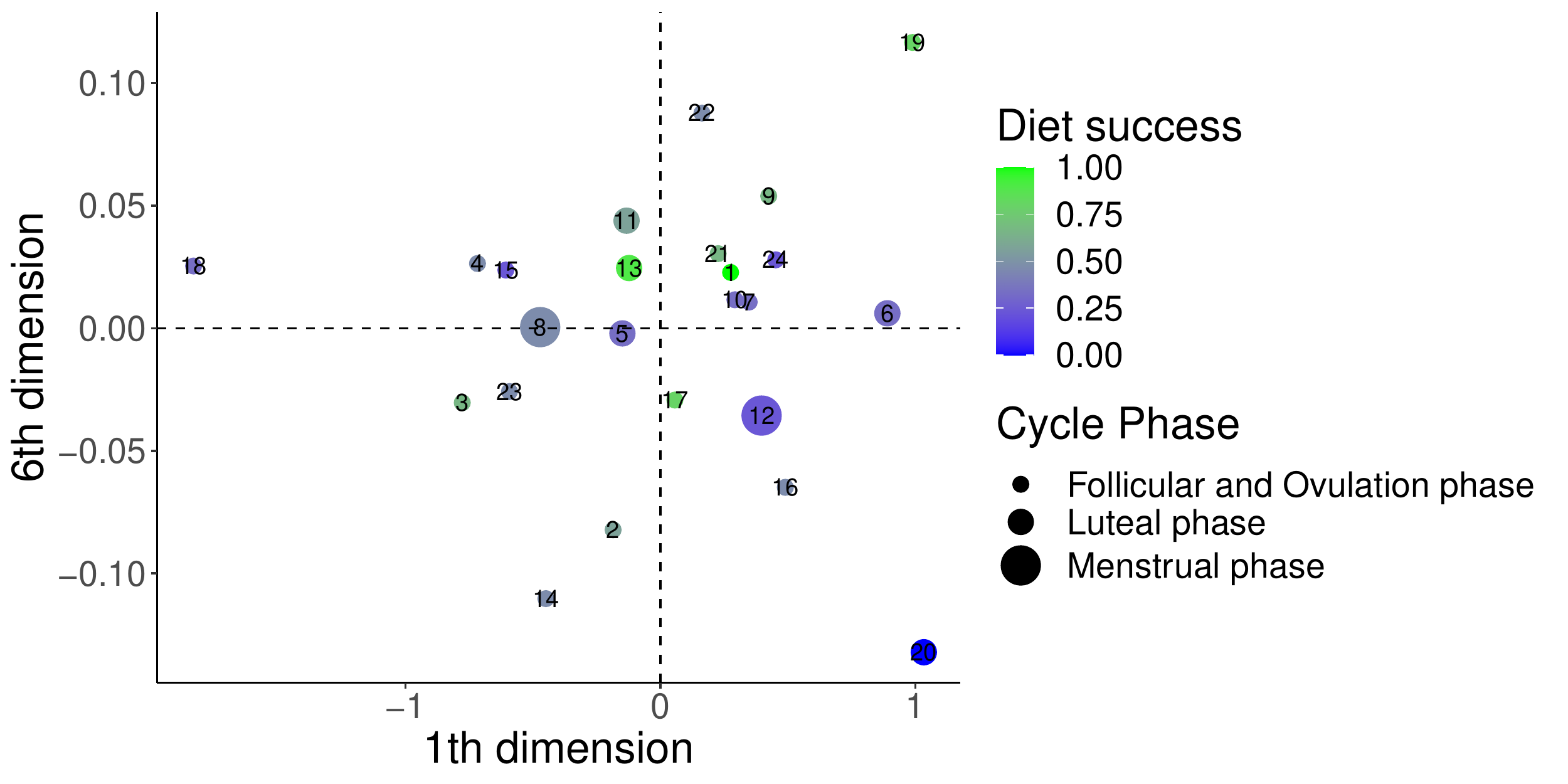}
\caption{$D_{raw}$ analysis: The $x$-axis represents the $1$th fitted configuration computed by the multidimensional scaling approach, while the $y$-axis shows the $6$th fitted configuration. The color gradient describes the appetite post covariate (scaled), while the size of the points specifies the phase of the cycle.}
\label{fig:distance_raw}
\end{figure}

To sum up, the \emph{rotational}-based distance $D_{Ro}$ allows capturing the functional variability in neural response in terms of \emph{rotational effort}. The subjects' neural activation then differs in terms of orientation. However, these differences are lessened after and before the application of functional alignment, i.e., considering the \emph{residual}-based distance $D_{Re}$ and the \emph{raw}-based distance $D_{raw}$ when some subject-specific covariates are analyzed in the same time. 

\section{Conclusions}
In this manuscript, we proposed \emph{Procrustes}-based distances based on the aligned images and orthogonal transformations estimated by \emph{Procrustes}-based methods. These distances permit the exploration of the dissimilarity between matrices from two independent points of view. The \emph{residual}-based distance expresses the dissimilarity in terms of functional columns net to rotations, i.e., eliminating the orientation component. Instead, the \emph{rotation}-based distance describes the dissimilarity in terms of functional (mis)alignment of the matrices' columns, i.e., how the matrices have similar column orientations. The method is helpful when the research aim is to analyze matrices expressed in an arbitrary coordinate system. In addition, the proposed distances can also be advantageous when the focus is exploring the distances between big data matrices, e.g., fMRI application. In this framework, each subject is represented by a vast matrix with approximately $300 \times 200,000$ dimensions. The \emph{Procrustes}-based distances permit the exploration of these matrices in a space with dimensions equal to the number of matrices/subjects analyzed. In the fMRI application, we found that these metrics result in reliable measures of individual differences. In fact, the \emph{Procrustes}-based functional alignments permit reducing confounds from topographic idiosyncrasies and capturing variation around shared functional and anatomical responses across individuals. The distances proposed in this manuscript allowed to find groups of individuals sharing patterns of neural brain activation. In conclusion, the \emph{Procrustes}-based distances thus add valuable exploratory and visualization tools to the world of Procrustes' methods.

\section*{Acknowledgment}
Angela Andreella gratefully acknowledges funding from the grant BIRD2020 /SCAR ASSEGNIBIRD2020\_01 of the University of Padova, Italy, and PON 2014-2020/DM 1062 of the Ca’ Foscari University of Venice, Italy. Some of the computational analyses done in this manuscript were carried out using the University of Padova Strategic Research Infrastructure Grant 2017: "CAPRI: Calcolo ad Alte Prestazioni per la Ricerca e l’Innovazione", \url{http://capri.dei.unipd.it}. 

\section*{Author contributions}
\textbf{Angela Andreella}: conceptualization, software, data curation, formal analysis, investigation, and writing of the original draft, review \& editing. \textbf{Riccardo De Santis}: Conceptualization, methodology, writing the original draft, review \& editing. \textbf{Anna Vesely}: Conceptualization, methodology, writing the original draft, review \& editing. \textbf{Livio Finos}: Conceptualization, methodology, writing the
original draft, and supervision.

\section*{Declaration of Competing Interest}
The authors declare no competing interests.

\newpage

\bibliographystyle{apalike}
\bibliography{bibliography.bib} 

\end{document}